\documentclass[11pt]{article}

\usepackage{amsmath}
\usepackage{amsfonts}
\usepackage{amssymb}
\usepackage{mathtools}
\usepackage[mathscr]{eucal}
              \newif{\ifcomentarios}\comentariosfalse\newtheorem{theorem}{Theorem}

\newtheorem{definition}[theorem]{Definition}
\newtheorem{lemma}[theorem]{Lemma}
\newtheorem{remark}[theorem]{Remark}
\newtheorem{corollary}[theorem]{Corollary}
\newtheorem{proposition}[theorem]{Proposition}
\newenvironment{proof}[1][Proof]{\textit{#1.} }{\hfill $\Box$}

\newcommand{\OBSI}{\begin{remark}\begin{rm}}
\newcommand{\OBSF}{\end{rm}\end{remark}}
\newcommand{\DEFI}{\begin{definition}\begin{rm}}
\newcommand{\DEFF}{\end{rm}\end{definition}}
\newcommand{\zerarcounters}{\setcounter{equation}{0}\setcounter{theorem}{0}}

\newcommand{\be}{\begin{eqnarray}}
\newcommand{\en}{\end{eqnarray}}
\newcommand{\bee}{\begin{eqnarray*}}
\newcommand{\ene}{\end{eqnarray*}}

\newcommand{\LL}{\mathrm{L}}
\newcommand{\CC}{\mathrm{C}}
\newcommand{\var}{\varepsilon}

\DeclareMathOperator*{\supp}{supp}

\DeclareMathOperator*{\esssup}{ess.sup}
\DeclareMathOperator*{\dom}{dom}

\begin{document}
\title{A characterization of singular packing subspaces with an application to  limit-periodic operators}
\author{Silas L. Carvalho \quad and \quad C\'esar  R. de Oliveira}

\date{\today}

\maketitle

\begin{abstract}
A new characterization of the singular packing subspaces of general bounded self-adjoint operators is presented, which is used to show that the set of operators whose spectral measures have upper packing dimension equal to one is a~$G_\delta$ (in suitable metric spaces). As an application, it is  proven that, generically (in space of continuous sampling functions), spectral measures of the limit-periodic Schr\"odinger operators have upper packing dimensions equal to one. Consequently, in a generic set, these operators are quasiballistic.
 \end{abstract}

\

\

\section{Introduction and results}
\label{INT}

\zerarcounters

A study of packing continuity and singularity of bounded self-adjoint operators is performed. Our main goals here are to present a dynamical characterization of singular $\alpha$-packing subspaces (Theorem~\ref{l21}), and  use this result (along with the well known equivalence between strong convergence and strong dynamical convergence of operators) to prove, for some metric space of self-adjoint operators, that the set of operators whose spectral measures have upper packing dimension equal to~1 is a~$G_\delta$ set (Theorem~\ref{prop1uPd}). The dynamical characterization will be obtained in Section~\ref{sectHW} through the notion of uniformly $\alpha$-H\"older singular measures and Lemma~\ref{Stri} (which is a ``singular version'' of Strichartz's Theorem~\cite{strichartz}).

To put our work into perspective, we mention that, although important, it is not always easy to present dynamical characterizations of spectral subspaces of a self-adjoint operator~$T$. For a vector~$\xi$, denote the so-called return probability at time~$t$ by $p_\xi(t):=\left|\langle\xi,e^{-itT}\xi\rangle\right|^2$, clearly a dynamical quantity. It is well known (see, for instance, Theorem~13.5.5 in~\cite{Cesar}) that the absolutely continuous subspace of~$T$ is the closure of the vectors~$\xi$ so that $p_\xi\in\LL^1(\mathbb R)$. By using Strichartz's Theorem, recalled in Section~\ref{sectHW}, Last (Theorem~5.3 in~\cite{Last}) has presented a dynamical characterization of the $\alpha$-Hausdorff continuous subspace~$\mathcal{H}_{\alpha\mathrm{Hc}}^T$ of~$T$ (the definition is similar to~$\mathcal{H}_{\alpha\mathrm{Pc}}^T$ in Proposition~\ref{DEC}) in terms of averages $\langle p_\xi\rangle(t):=\frac1t\int_0^t p_\xi(s)\,{\mathrm d}s$, that is, given  $0<\alpha<1$, then for all~$\varepsilon>0$,
\[
\mathrm{closure}\left\{\xi\mid \sup_t t^{\alpha+\varepsilon}  \langle p_\xi\rangle(t)<\infty  \right\}\subset \mathcal{H}_{\alpha\mathrm{Hc}}^T\subset \mathrm{closure}\left\{\xi\mid \sup_t t^{\alpha}  \langle p_\xi\rangle(t)<\infty  \right\}.
\]
 Here we have a parallel of this result (Theorem~\ref{l21}) for the $\alpha$-packing singular subspace; for this, we have introduced an appropriate quantity in equation~\eqref{DefL} and the concept of U$\alpha$HS measures (see Definition~\ref{HoS}).

We apply our general packing results to a class of limit-periodic operators. These are discrete one-dimensional ergodic Schr\"odinger operators, denoted by $H^\kappa_{g,\tau}$, acting in $l^2(\mathbb{Z})$, whose action is given by
\begin{equation}\label{LPo}
(H^\kappa_{f,\sigma}\psi)_n=\psi_{n+1}+\psi_{n-1}+V_n(\kappa)\psi_n\;,
\end{equation}
\noindent with 
\begin{equation}\label{VF}
V_n(\kappa)=g(\tau^n(\kappa))\,;
\end{equation} \noindent here, $\kappa$ belongs to a Cantor group~$\Omega$, $\tau:\Omega\rightarrow\Omega$ is a minimal translation on~$\Omega$ and $g:\Omega\rightarrow\mathbb{R}$ a continuous sampling function, i.e., $g\in\CC(\Omega,\mathbb R)$ with the norm of uniform convergence. For more details, see~\cite{Avila}.

For each~$\kappa\in\Omega$, let $X_{\kappa}$ be the set of limit-periodic operators $H_{g,\tau}^\kappa$ given by \eqref{LPo} and \eqref{VF}, with metric 
\begin{equation}\label{metric3}
d(H_{g,\tau}^\kappa,H_{g',\tau}^\kappa)=\|g-g'\|_\infty\;.
\end{equation} We shall prove the following

\begin{theorem}\label{LPO}
 For each $\kappa\in\Omega$, the set~$C_{\mathrm{1uPd}}^\kappa:=\{T\in X_\kappa\mid \sigma(T)$ is purely $1$-upper packing dimensional$\}$ is generic in~${X}_{\kappa}$.
\end{theorem}

We stress that in a recent work~\cite{SCC}, it was proven, again for some metric spaces of self-adjoint operators, that the set of operators whose spectral measures have upper correlation dimension equal to 1 is a~$G_\delta$. Since every  Borel and finite measure on~$\mathbb{R}$ whose upper correlation dimension is one has upper packing dimension equal to 1, one can conclude that if the hypotheses in Theorem 4.1 in~\cite{SCC} are fulfilled, then Theorem~\ref{prop1uPd} follows. However, the hypotheses in Theorem~\ref{prop1uPd} below are weaker than in Theorem 4.1 in~\cite{SCC}, and therefore  easier to meet. In particular, we were not able to apply the method discussed in~\cite{SCC} to the class~\eqref{LPo} of limit periodic operators, since the estimation of upper correlation dimension seems to be far from trivial in this case.

The notorious Wonderland theorem~\cite{Simon} gives sufficient conditions for a set of self-adjoint operators, whose spectrum is purely singular continuous, to be generic. Theorem~\ref{prop1uPd} is another step towards a better comprehension of the typical (in a topological sense) behavior of the spectral measures of a self-adjoint operator, since it generalizes Wonderland theorem in the sense that every application of Wonderland theorem discussed in~\cite{Simon} for bounded operators can be extended to results about generic sets of operators whose spectral measures have upper packing dimension equal to 1 (on top of singular continuous spectrum).

The proof of Theorem~\ref{LPO} is presented in Section~\ref{PLPO}, since a previous preparation is required. Let us briefly discuss some dynamical consequences of Theorem~\ref{LPO} within the context of the unitary evolution group $e^{-itT}$, that is, the solution to the corresponding Schr\"odinger equation~\cite{Cesar}. Regarding the dynamics generated by a self-adjoint operator~$T$ acting on~$l^2(\mathbb{Z})$,  the  growth of the width of ``quantum wave packets'' is usually probed by the algebraic growth of the (time-averaged) $q$-moments, $q>0$, 
\[
\langle M_{T}^q\rangle(t):=\sum_{n}|n|^q \,\frac2t\int_0^\infty e^{-2s/t} |\langle
e^{-isT}\delta_0,\delta_n\rangle|^2\,\mathrm{d}s\,,
\]  of the position operator at time~$t>0$; such wave packets are represented here by the initial state~$\delta_0$. To  describe this algebraic growth $\langle M_{T}^q\rangle(t)\sim t^{q\beta(q)}$ for large~$t$, one usually considers the lower and the upper transport exponents, given respectively by
\[
\beta_{T}^-(q):=\liminf_{t\rightarrow\infty}\frac{\ln \langle M_{T}^q\rangle(t)}{q\ln t},\quad 
\beta_{T}^+(q):=\limsup_{t\rightarrow\infty}\frac{\ln \langle M_{T}^q\rangle(t)}{q\ln t}.
\]

The following result, extracted from~\cite{GK}, gives basic properties of moments within the setting of bounded self-adjoint operators~$T$ acting on~$l^2(\mathbb{Z})$, in particular, for bounded discrete Schr\"odinger operators (that is, operators whose action is given by~\eqref{LPo} with  bounded potentials~$(V_n)$).

\begin{proposition} \label{GK} If~$T$ is a bounded self-adjoint operator on~$l^2(\mathbb Z)$, then 
\begin{enumerate}
\item $\langle M_{T}^q\rangle(t)$ is well defined for all $q,t>0$;
\item $\beta_{T}^\pm(q)$ are increasing functions of $q$;
\item $\beta_{T}^\pm(q)\in[0,1]$, for all $q>0$.
\end{enumerate}
\end{proposition}

In case $\beta_{T}^+(q)=1$ (resp.\ $\beta_{T}^-(q)=1$), for all $q>0$, the corresponding dynamics is called {\em quasi-ballistic} (resp.\ {\em ballistic}). We shall make use of the general inequality (see Definition~\ref{defDimMu} for the description of the upper packing dimension $\dim^+_{\mathrm P}(\cdot)$)
\begin{equation}\label{bLIP}
 \beta_{T}^+(q) \ge \dim^+_{\mathrm P}\left(\mu^T_{\delta_0}\right)\,,\qquad \forall q>0,
\end{equation} 
proven in~\cite{Guar}. It is worth mentioning that  $ \beta_{T}^-(q)$ is related to the upper Hausdorff dimension~\cite{Guar1,BCM}, but we will not use them in this work. Following the discussion above and Theorem~\ref{LPO}, one has

\begin{corollary}\label{AMD1} 
For every~$\kappa\in\Omega$, the set~$C_{{\mathrm{QB}}}^\kappa:=\{T\in X_\kappa\mid \beta_{T}^+(q)=1$ for every $q>0\}$ is generic in~$X_\kappa$.
\end{corollary} 

In Section~\ref{HPM}, we recall suitable results on decompositions of Borel measures with respect to packing measures and dimensions, and how these components are related to pointwise scaling exponents and generalized dimensions. In Section~\ref{sectHW}, we present a singular version-like of Strichartz's Theorem~\cite{strichartz,Last} for finite measures, which is used to give a dynamical characterization of packing singular subspaces of  bounded self-adjoint operators in general separable Hilbert spaces  (see Theorem~\ref{l21}). The very same arguments lead to a version for spectral measures restricted to any given subinterval of the real line. 

In Section~\ref{GenQB}, we state and prove Theorem~\ref{prop1uPd}. In the last section, as previously stated, we present the proof of Theorem~\ref{LPO}. It is important to emphasize that the results of Section~\ref{GenQB} can be used to prove, for every general class of bounded operators (including classes of not necessarily Schr\"odinger-like operators) discussed in~\cite{Simon}, existence of generic sets of operators whose spectral measures have upper packing dimensions equal to 1.

Some words about notation: $\mathcal H$ will always denote a complex separable Hilbert space. The spectrum of a self-adjoint operator~$T$ is denoted by~$\sigma(T)$. $\mu$ will always indicate, unless explicitly stated, a finite nonnegative Borel measure on~$\mathbb R$, and its restriction to a Borel set~$E$ will be denoted by~$\mu_{;E}$;  it is \textit{singular} if~$\mu$ and the Lebesgue measure are mutually singular; it is \textit{supported} on a Borel set~$S$ if $\mu(\mathbb R\setminus S)=0$. $\supp(\mu)$ denotes the \textit{support} of~$\mu$, that is, the complement of the largest open set~$B$ with~$\mu(B)=0$.  $P^T(E)$ represents the spectral projection of~$T$ associated with the Borel set~$E\subset\mathbb{R}$.

\section{Basic tools} \label{HPM}   
\zerarcounters

In this section, we recall important decompositions of Borel measures on~$\mathbb R$ with respect to packing measures and dimensions, along with the corresponding spectral decompositions of self-adjoint operators. We also recollect how these decompositions are related to the upper and generalized dimensions. This discussion parallels the rather well-known corresponding Hausdorff properties. 

\subsection{Packing decompositions}
\label{HPM1}

Given a set~$S\subset\mathbb R$ and $0\le\alpha\le1$, denote by $h^\alpha(S)$ its $\alpha$-{dimensional (exterior)  Hausdorff measure} and by~$\dim_{\mathrm H}(S)$ its Hausdorff dimension. Since packing measures and dimensions are not so familiar as the Hausdorff versions, we recall their definitions in what follows. 

A $\delta$-packing of an arbitrary set $S\subset\mathbb{R}$ is a countable disjoint collection $(\bar B(x_k;r_k))_{k\in\mathbb{N}}$ of closed intervals centered at~$x_k\in S$ and radii $r_k\le\delta/2$, so with diameters at most of~$\delta$. Define $P^\alpha_\delta(S)$, $0\le \alpha\le1$, as
\begin{equation*}
 P^\alpha_\delta(S)=\sup\Big\{\sum_{k=1}^\infty(2r_k)^\alpha\mid (\bar B(x_k;r_k))_k \;\mathrm{is\; a} \;\delta\text{-packing\;of}\;S\Big\};
\end{equation*}
that is,  the supremum is taken over all~$\delta$-packings of~$S$. Then, take the decreasing limit
\begin{equation*}
 P_0^\alpha(S)=\lim_{\delta \downarrow 0}P^\alpha_\delta(S)
\end{equation*} as  a pre-measure.
\DEFI
The \textit{$\alpha$-packing (exterior) measure} $P^\alpha(S)$ of~$S$ is given by
\begin{equation*}
 P^\alpha(S):=\inf\Big\{\sum_{k=1}^\infty P_0^\alpha(S_k)\mid S\subset\bigcup_{k=1}^\infty S_k \Big\} \;. 
\end{equation*}\label{DPack}
\DEFF

The  \textit{packing dimension} of the set~$S$,~$\dim_{\mathrm P}(S)$, is defined  as the infimum of all~$\alpha$ such that  $P^\alpha(S)=0$, which coincides with the supremum of all~$\alpha$ so that~$P^\alpha(S)=\infty$.

It is possible to show~\cite{F} that the Hausdorff and packing dimensions 
are related by the inequality $\dim_{\mathrm H}(S)\le\dim_{\mathrm P}(S)$, and this 
inequality is in general strict. It is also important to mention that $P^\alpha$ and $h^\alpha$ are Borel (regular) measures and,  for~$0\le\alpha<1$, they are not $\sigma$-finite; furthermore,  $P^0\equiv h^0$ and $P^1\equiv h^1$, and they are equivalent, respectively, to  counting and Lebesgue measures. 

\DEFI\label{defDimMu}
 The packing {\em upper dimension} of~$\mu$  is defined as 
\[
\dim_{\mathrm P}^+(\mu):=\inf \{ \dim_{\mathrm P}(S)\mid \mu(\mathbb R\setminus S)=0, \;S\; \mathrm{a\; Borel\; subset\; of}\; \mathbb R \}.
\]
\DEFF

The notions of packing measures and dimensions lead to concepts of continuity and singularity of Borel measures with respect to them. 

\DEFI\label{SCB}
Let $\alpha \in [0,1]$.  $\mu$ is called:
\begin{enumerate}
\item  \textit{$\alpha$-packing continuous}, denoted $\alpha\mathrm{Pc}$, if~$\mu (S)=0$ for every Borel set~$S$ with $\mathrm{P}^{\alpha}(S)=0$.
\item  \textit{$\alpha$-packing singular}, denoted $\alpha\mathrm{Ps}$, if it is supported on some Borel set~$S$  with $\mathrm{P}^{\alpha}(S)=0$. 
\item  \textit{$\alpha$-packing dimension continuous}, denoted $\alpha\mathrm{Pdc}$, if~$\mu (S)=0$ for every Borel set~$S$ with $\dim_{\mathrm{P}}(S)<\alpha$.
\item  \textit{almost~$\alpha$-packing dimension singular}, denoted $\mathrm{a}\alpha\mathrm{Pds}$, if it is supported on some Borel set~$S$ with $\dim_{\mathrm{P}}(S)\le\alpha$.
\item 1-\textit{packing dimensional}, denoted $1\mathrm{Pd}$, if $\mu(S)=0$ for any Borel set $S$ with $\dim_{\mathrm P}(S)<1$.
\end{enumerate}
\DEFF

\begin{proposition}\label{AHC}
Let~$\mu$ as before. 
\begin{enumerate}
\item Fix $\alpha\in(0,1]$. If~$\mu$ is $\alpha\mathrm{Ps}$, then it is $\mathrm{a}\alpha\mathrm{Pds}$. 
\item Fix $\alpha\in[0,1)$. If~$\mu$ is $\mathrm{a}\alpha\mathrm{Pds}$, then it is $(\alpha+\varepsilon)\mathrm{Ps}$ for every $0<\varepsilon\le 1-\alpha$.
\end{enumerate}
\end{proposition}

 \begin{proposition}\label{DEC}
Let $T:\dom T\subset \mathcal{H}\rightarrow\mathcal{H}$ be a self-adjoint operator in the Hilbert space~$\mathcal H$, and  $\mu_\psi^T$ the spectral measure of~$T$ associated with the vector~$\psi\in\mathcal H$. Given $\alpha\in(0,1)$, the sets 
\begin{eqnarray*}
\mathcal{H}_{\alpha\mathrm{Pc}}^T:=\left\{\psi\mid \mu_\psi^T\; \mathrm{is} \; \alpha\mathrm{Pc}\right\} \quad  \mathrm{and} \quad \mathcal{H}_{\alpha\mathrm{Ps}}^T:=\left\{\psi\mid \mu_\psi^T\; \mathrm{is}\; \alpha\mathrm{Ps}\right\}
 \end{eqnarray*}  are closed and mutually orthogonal subspaces of~$\mathcal H$, which are invariant under~$T$, and $\mathcal{H}=\mathcal{H}_{\alpha\mathrm{Pc}}^T\oplus\mathcal{H}_{\alpha\mathrm{Ps}}^T$.
 \end{proposition}
 \begin{proof}
 The proof follows closely the proofs of the corresponding statements involving Hausdorff versions in Theorem~5.1 in~\cite{Last}.   
 \end{proof}

\subsection{Generalized upper dimensions}\label{subsectPSE}

\DEFI \label{LD}
 The {\em pointwise upper scaling exponent}  of~$\mu$ at $x\in\mathbb{R}$ is defined as
\[
d_{\mu}^{+}(x):=\limsup_{\varepsilon\rightarrow 0}\frac{\ln \mu(B(x;\varepsilon))}{\ln\varepsilon}\,,
\] if, for every $\var>0$, $\mu(B(x;\var))>0$, and $d_{\mu}^{+}(x):=+\infty$ otherwise. 
\DEFF

The function $x\mapsto d_{\mu}^{+}(x)$ is measurable~\cite{F}. The following results relate packing singularity properties of nonnegative finite Borel measures on~$\mathbb{R}$ with their upper pointwise scaling exponent and  dimensions (see~\cite{Guar} for details).

\begin{proposition}\label{Pacbs}
Let~$\mu$ as before.
\begin{enumerate}
\item If~$\mu$ is $\alpha\mathrm{Ps}$, then $\mu\mathrm{\text{-}}\esssup d_{\mu}^{+}\le\alpha\le 1$;
\item~$\mu$ is $\mathrm{a}\alpha\mathrm{Pds}$ if, and only if, $\mu\mathrm{\text{-}}\esssup d_{\mu}^{+}\le\alpha\le 1$.
\end{enumerate}
\end{proposition}

 \begin{proposition}\label{PRo1}
 Let~$E$ be a Borel subset of~$\mathbb{R}$. Then,
 \begin{equation*}
 \dim_{\mathrm P}^+(\mu_{;E})=\mu_{;E}\mathrm{\text{-}}\esssup d_{\mu_{;E}}^{+}\;.
 \end{equation*}
 \end{proposition}

Item~2.\ in Proposition~\ref{Pacbs} can be restated in the following way:

\begin{corollary}
\label{CPRo1}
Let~$E$ be a Borel subset of~$\mathbb{R}$. $\mu_{;E}$ is $\mathrm{a}\alpha\mathrm{Pds}$ if, and only if, $\dim_{\mathrm P}^+(\mu_{;E})\le\alpha$.
\end{corollary}

Now we recall the definition of generalized upper dimensions of positive Borel measures and how they are connected to the upper packing dimensions.

\DEFI\label{DD2}
Let~$\mu$ be a positive Borel measure on~$\mathbb{R}$. The {\em upper generalized dimensions} of~$\mu$ are defined, for $q\neq 1$, as
\[ D_q^{+}(\mu):=\limsup_{\varepsilon\rightarrow 0}\frac{\ln\left[\int[\mu(B(x;\varepsilon))]^{q-1}\,\mathrm{d}\mu(x)\right]}{(q-1)\ln\varepsilon}\;,
\]
with integrals taken on~$\supp\mu$.
\DEFF

For all $q<1<s$,  Proposition~4.1 in~\cite{BGT3}  gives
\begin{equation}\label{D2+}
D_q^{+}(\mu)\ge\dim_{\mathrm P}^+(\mu)\ge D_s^{+}(\mu)\,. 
\end{equation}   This will be used in Section~\ref{sectHW}, particularly with~$q=1/2$.

\section{Dynamical characterization of $\mathcal{H}_{\alpha\text{Ps}}^T$}\label{sectHW}
\zerarcounters

 In Definition~\ref{HoS}, we introduce a class of special measures for which a singular version-like of Strichartz's Theorem (Theorem~\ref{STRI}) will be deduced (see Lemma~\ref{Stri}). The  arguments there will be used to prove the main result of this section, that is, Theorem~\ref{l21}. We assume, in what follows, that~$0\ne T$ represents a bounded self-adjoint operator  on~$\mathcal{H}$.

\begin{definition}\label{HoS}
Let $\alpha\in[0,1]$ and~$\mu$ be a positive Borel measure on~$\mathbb{R}$. We say that~$\mu$ is \it{uniformly $\alpha$-H\"older singular} (U$\alpha$HS) if there exist positive constants~$C$ and~$r_0$, with $r_0<1$, such that, for all $0<r<r_0$ and for~$\mu$ a.e.~$x$, $\mu(B(x;r))\ge Cr^\alpha$.
\end{definition}

Besides the application to limit-periodic operators ahead, the next results have interest on their own (as well as the results in Section~\ref{GenQB}). Next a description of $\alpha$-packing singular measures in terms of U$\alpha$HS ones.

\begin{theorem} Let~$\mu$ be a positive Borel measure on~$\mathbb{R}$ and $\alpha\in(0,1)$. If~$\mu$ is $\alpha$Ps, then, for every $0<\var\le 1-\alpha$ and every $\delta>0$, there exist a Borel set~$S_\delta=S\subset\mathbb{R}$ such that~$\mu(S^c)<\delta$ and  positive constants~$C$ and~$r_0<1$ such that, for each $0<r<r_0$ and for each~$x\in S$, $\mu(B(x;r))\ge Cr^\alpha$. Conversely, if, for every $\delta>0$, there exist mutually singular Borel measures $\mu_1^\delta$ and $\mu_2^\delta$ such that $\mu=\mu_1^\delta+\mu_2^\delta$, with $\mu_1^\delta$ U$\alpha$HS and $\mu_2^\delta(\mathbb{R})<\delta$, then, for every $0<\var\le 1-\alpha$,~$\mu$ is ($\alpha+\var$)Ps. \label{THoS}
\end{theorem}
\begin{proof}
Suppose that, for every $\delta>0$, $\mu=\mu_1^\delta+\mu_2^\delta$, with $\mu_1^\delta$ and $\mu_2^\delta$ satisfying the properties in the statement of the theorem. We must show, for every $0<\var\le1-\alpha$, that~$\mu$ is ($\alpha+\var$)Ps; by Propositions~\ref{AHC} and~\ref{Pacbs}, this is equivalent to showing that $\mu\text{-}\esssup d_\mu^+\le\alpha$. 

Let us assume, nonetheless, that $\mu\text{-}\esssup d_\mu^+>\alpha$. Thus, there is a Borel set, say $B$, of positive~$\mu$-measure such that $d_\mu^+(x)>\alpha$ for every $x\in B$. Fix $0<\zeta<\mu(B)$. By  hypothesis, there is a Borel set $E$ (which may depend on $\zeta$) such that~$\mu$ can be decomposed as $\mu=\mu_1^\zeta+\mu_2^\zeta$, with $\mu_1^\zeta(\cdot):=\mu(E\cap\cdot)$ U$\alpha$HS and $\mu_2^\zeta(\cdot):=\mu(E^c\cap\cdot)$, with $\mu_2^\zeta(\mathbb{R})<\zeta$.

By Definition~\ref{HoS}, there are constants $C>0$ and $0<r_0<1$ such that, for every $0<r<r_0$ and every $x\in E\setminus D$ ($D$ a set of zero~$\mu_1^\zeta$-measure), $\mu_1^\zeta(B(x;r))\ge Cr^{\alpha}$. Now, since $\ln\mu(\cdot)\ge \ln\mu_1^\zeta(\cdot)$, 
\[
d_\mu^+(x)=\limsup_{r\downarrow 0}\frac{\ln\mu(B(x;r))}{\ln r}\le\limsup_{r\downarrow 0}\frac{\ln\mu_1^\zeta(B(x;r))}{\ln   r}\le\limsup_{r\downarrow 0}\frac{\ln C}{\ln r}+\alpha=\alpha,
\] and $d_\mu^+(x)\le\alpha$ for every~$x\in E\setminus D$.  But then, $(E\setminus D)^c\supset B$, which implies that $\zeta>\mu_2^\zeta(E^c\cup D)=\mu_2^\zeta(E^c\cup D)+\mu_1^\zeta(E^c\cup D)=\mu(E^c\cup D)\ge\mu(B)$, a contradiction with $\mu(B)>\zeta$. Thus, $\mu\text{-}\esssup d_\mu^+\le\alpha$, and we are done.

Conversely, if~$\mu$ is $\alpha$Ps, then, by Proposition~\ref{Pacbs}, $\mu\text{-}\esssup d_\mu^+\le\alpha$; that is, 
\[
\limsup_{r\downarrow 0}\frac{\ln\mu(B(x;r))}{\ln r}\le\alpha
\] for~$\mu$ a.e.~$x$. Since the sequence $(f_r(x))$ of measurable functions \[
f_r(x):=\sup_{r^\prime\le r}\frac{\ln\mu(B(x;r^\prime))}{\ln r^\prime}
\] converges to $d_\mu^+(x)$, Egoroff's Theorem implies that given an arbitrary $\delta>0$, there is a Borel set~$S$ such that $\mu(S^c)<\delta$ and $f_r(x)$ converges uniformly on $S$ to $d_\mu^+(x)$, as $r\downarrow 0$. But then, given an arbitrary $0<\var\le1-\alpha$, there is a $0<r_0<1$ such that, for every $0<r<r_0$ and all~$x\in S$, $\ln\mu(B(x;r))/\ln r<\alpha+\var$; that is, $\mu(B(x;r))>r^{\alpha+\var}$, for  all~$x\in S$. 
\end{proof}

Now we introduce another quantity that has proven useful. For a finite and positive Borel measure~$\mu$ on~$\mathbb R$ and every $t\in\mathbb{R}$, write  
\begin{equation}\label{DefL} \Xi_{\mu}(t):=\int {\mathrm d}\mu(x)\left(\int {\mathrm d}\mu(y)\,e^{-(x-y)^2t^2/4}\right)^{-1/2}.
\end{equation} If the measure~$\mu$ is a spectral measure $\mu_\psi^T$, we denote $ \Xi_{\mu}(t)$ by~$\Xi_{\psi}^T(t)$.

Recall that~$\mu$ is uniformly $\alpha$-H\"older continuous~\cite{Last} if there  are positive and finite constants~$C$ and~$r_0$, so that for each $0<r<r_0$ and for~$\mu$ a.e.~$x$, $\mu(B(x;r))\le C r^\alpha$. The following result is known as Strichartz's Theorem (it is, in fact, an adapted version of the Theorem presented in~\cite{Last} for~$f\equiv 1\in \LL^2(\mathbb{R};{\mathrm d}\mu)$; see also~\cite{strichartz} for the original result).

\begin{theorem}\label{STRI}
Let~$\mu$ be a finite uniformly $\alpha$-H\"older continuous measure, and for each~$s>0$, denote
\[
\widehat{\mu}(s):=\int e^{-isx}\,{\mathrm d}\mu(x).
\]
Then, there exist constants~$\tilde{D}$ and~$t_0>0$, depending only on~$\mu$, such that, for any~$t>t_0$,
\begin{equation}\label{eq:returnProb}
\frac{1}{t}\int_0^t \vert\widehat{\mu}(s)\vert^2\, {\mathrm d}s\le \tilde{D} t^{-\alpha}.
\end{equation} 
\end{theorem}

 The proof of Theorem~\ref{STRI} in~\cite{Last}, after some preparation, essentially consists of showing that there exist constants $\tilde{D}$ and $t_0>0$ so that
\begin{equation}\label{eqAHc}
\int {\mathrm d}\mu(x)\left(\int {\mathrm d}\mu(y)\,e^{-(x-y)^2t^2/4}\right)\le \tilde D/t^{\alpha},\qquad t>t_0.
\end{equation}
In such proofs, in case $\mu=\mu^T_\psi$, the parameter $``t"$ comes from the time evolution $e^{-itT}\psi$ and the left hand side of~\eqref{eq:returnProb} coincides with the average return probability $\langle p_\psi\rangle(t)$, so  one may look at~$\Xi_\mu(t)$ as a dynamical quantity. Equation~\eqref{eqAHc}, related to  Hausdorff continuity, was our main motivation to introduce~$\Xi_\mu(t)$, and we have got the following singular version of this result; although simple, it will be very useful ahead.

\begin{lemma}\label{Stri}
Let~$\mu$ be a finite positive Borel measure on~$\mathbb R$ and U$\alpha$HS for some $\alpha\in[0,1]$. Then, there exist finite constants $D>0$ and $t_0>1$  such that, for every  $t>t_0$, 
\[ \Xi_{\mu}(t)\le\mu(\mathbb{R})\, D\, t^{\alpha/2}\,.
\] In case of spectral measures $\mu_\psi^T$, one has  $\Xi^T_{\psi}(t)\le\Vert\psi\Vert^2\, D\, t^{\alpha/2}.
$\end{lemma}
\begin{proof} Since~$\mu$ is U$\alpha$HS, there are positive constants~$C$ and~$r_0$, with $r_0<1$, such that, for every $0<r<r_0$ and~$\mu$ a.e.~$y$, $\mu(B(y;r))\ge Cr^\alpha$. Thus, by taking $t_0\equiv 1/r_0$, it follows, for every $t>t_0$ and every~$x\in\mathbb{R}$, that 
\begin{eqnarray}
\nonumber \int {\mathrm d}\mu(y)\,e^{-(x-y)^2t^2/4}&=&\sum_{n\ge 0}\int_{n/t\le|x-y|<(n+1)/t}{\mathrm d}\mu(y)\,e^{-(x-y)^2t^2/4}\\
&\ge& 2Ct^{-\alpha}\sum_{n\ge 0}\,e^{-(n+1)^2/4}\,.\label{eqXi}
\end{eqnarray}

Finally, by letting $D\equiv \left(2C\sum_{n\ge 0}e^{-(n+1)^2/4}\right)^{-1/2}$, we obtain 
\[\Xi_{\mu}(t)=\int {\mathrm d}\mu(x)\left(\int {\mathrm d}\mu(y)\,e^{-(x-y)^2t^2/4}\right)^{-1/2}\le\mu(\mathbb{R})\, D\,t^{\alpha/2}\,.
\]
In case of~$\mu=\mu_\psi^T$, just recall that $\mu_\psi^T(\mathbb R)=\|\psi\|^2$.  
 \end{proof}

\begin{proposition}\label{Pa1}
Let~$T$ be a bounded self-adjoint operator on~$\mathcal H$ and $\alpha\in(0,1)$. Then, for every $0<\var\le1-\alpha$, 
\[\mathcal{H}^T_{\alpha\mathrm{Ps}}\setminus\{0\}\subset\{\psi\mid\limsup_{t\rightarrow\infty}t^{-(\alpha+\var)/2}\,\Xi_{\psi}^T(t)<\infty\}.
\] 
\end{proposition}
\begin{proof}
Suppose that $\psi\in\mathcal{H}^T_{\alpha\textrm{Ps}}\setminus\{0\}$, that is, that $\mu_\psi^T$ is positive and~$\alpha$Ps. By Theorem~\ref{THoS}, it follows, for every $0<\varepsilon\le1-\alpha$ and every $\delta>0$, that there exist~$S\subset\mathbb{R}$ such that~$\mu(S^c)<\delta$ and positive constants~$C$ and~$r_0<1$ such that, for each $0<r<r_0$ and for each~$x\in S$, $\mu(B(x;r))\ge Cr^\alpha$. Since $e^{-(x-y)^2t^2/4}$ is positive, one has, for every $x,t\in\mathbb{R}$,
\[0<\int_S {\mathrm d}\mu_{\psi}^T(y)\,e^{-(x-y)^2t^2/4}\le\int {\mathrm d}\mu_{\psi}^T(y)\,e^{-(x-y)^2t^2/4}<\infty.
\] Thus, using the same reasoning presented in the proof of Lemma~\ref{Stri}, one has
\begin{eqnarray*} \Xi_{\psi}^T(t)
\le \int {\mathrm d}\mu_\psi^T(x)\left(\int_S {\mathrm d}\mu_{\psi}^T(y)\,e^{-(x-y)^2t^2/4}\right)^{-1/2}\le \Vert\psi\Vert^2\,  D\, t^{(\alpha+\var)/2},
\end{eqnarray*}
for some finite~$D$ and large~$t$; relation~\eqref{eqXi} and the identity~$\mu_\psi^T(\mathbb R)=\Vert\psi\Vert^2$ were used in the last inequality.   
 \end{proof}

\begin{lemma}\label{T2.2}
 Let~$T$ be a bounded self-adjoint operator on~$\mathcal{H}$ and $\alpha\in (0,1)$. Then, for each $0<\varepsilon\le\min\{\alpha,1-\alpha\}$, one has
\begin{eqnarray*}
\nonumber\big\{\psi\mid \limsup_{t\rightarrow\infty} t^{-(\alpha-\varepsilon)/2}\, \Xi_{\psi}^T(t)<\infty\big\}&\subset& \big\{\psi\mid D_{1/2}^{+}(\mu_{\psi}^T)\le\alpha-\var/2\big\}\\
&\subset&\big\{\psi\mid \limsup_{t\rightarrow\infty} t^{-(\alpha+\var)/2}\,\Xi_{\psi}^T(t)<\infty\big\}.
\end{eqnarray*}
\end{lemma}
\begin{proof}
 Since, by hypothesis,~$T$ is bounded,~$\mu_{\psi}^T$ has compact support. Hence, the result is immediate from   
\begin{equation}\label{egkt}
\limsup_{t\rightarrow\infty}\frac{\ln \Xi_{\psi}^T(t)}{\ln t}=\dfrac{1}{2}D_{1/2}^{+}(\mu_{\psi}^T)\;,
\end{equation} 
proved in Lemma~4.3 in~\cite{BGT} (note the different notation for~$\Xi_\psi^T$ in~\cite{BGT}).   
 \end{proof}
 
 \OBSI
The hypothesis that  the operator~$T$ is bounded can be dropped as soon as one verifies~\eqref{egkt}. See~\cite{BGT} for a discussion about the validity of~\eqref{egkt} in more general situations.
 \OBSF

\begin{proposition}\label{Pa2}
Let $\alpha\in(0,1)$. Then, for every $0<\var\le\alpha$, 
\[
\big\{\psi\mid\limsup_{t\rightarrow\infty}t^{-(\alpha-\var)/2}\,\Xi_\psi^T(t)<\infty\big\}\subset\mathcal{H}^T_{\alpha\mathrm{Ps}}\setminus\{0\}.
\]
\end{proposition}
\begin{proof}
Fix $0<\var\le\alpha$. If $\psi$ is such that $\limsup_{t\rightarrow\infty}t^{-(\alpha-\var)/2}\,\Xi_\psi^T(t)<\infty$, one has, from Lemma~\ref{T2.2}, that $D^+_{1/2}(\mu_\psi^T)\le\alpha-\var/2$. By inequalities~\eqref{D2+}, $D^+_{1/2}(\mu_\psi^T)\ge\dim_{\mathrm{P}}^+(\mu_\psi^T)$, and therefore, $\dim_{\mathrm{P}}^+(\mu_\psi^T)\le\alpha-\var/2$; consequently, it follows from Corollary~\ref{CPRo1} that $\mu_\psi^T$ is a($\alpha-\var/2$)Pds and so it is $\alpha$Ps, by Proposition~\ref{AHC}.  
\end{proof} 

By combining Propositions~\ref{Pa1} and~\ref{Pa2}, we obtain the following characterization of the $\alpha$-packing singular subspace.

\begin{theorem}\label{l21} 
Let~$T$ be a bounded self-adjoint operator on~$\mathcal{H}$. If $\alpha\in(0,1)$, then, for every $0<\var\le\min\{\alpha,1-\alpha\}$, 
\[
\Big\{\psi\mid\limsup_{t\rightarrow\infty}t^{-(\alpha-\var)/2}\,\Xi_\psi^T(t)<\infty\Big\}\subset\mathcal{H}^T_{\alpha\mathrm{Ps}}\setminus\{0\}\subset\Big\{\psi\mid\limsup_{t\rightarrow\infty}t^{-(\alpha+\var)/2}\,\Xi_\psi^T(t)<\infty\Big\}.
\]
\end{theorem}

\section{Generic quasiballistic and upper packing sets}\label{GenQB}
\zerarcounters

Let $(X,d)$ be a complete metric space of bounded self-adjoint operators, acting on the infinite-dimensional  separable Hilbert space~$\mathcal{H}$, such that the metric~$d$ convergence implies strong  convergence of operators. We denote its elements by~$T$. In order to obtain the main results of this section (i.e., Propositions~\ref{R1} and~\ref{R2C}), we will prove, for each fixed vector~$\psi\in\mathcal{H}$, that the set 
\[
C^{\psi;(a,b)}_{1\mathrm{uPd}}:=\Big\{T\in X\mid \dim_{\mathrm{P}}^+(\mu_{\psi;(a,b)}^T)=1\Big\}
\] is a $G_\delta$ set in~$X$.  Recall that here, $\mu_{\psi;(a,b)}^T$ denotes the restriction of $\mu_{\psi}^T$ to the open interval $(a,b)$, $-\infty\le a<b\le +\infty$. If~$(a,b)=\mathbb{R}$, we simply denote~$C^{\psi;(a,b)}_{1\mathrm{uPd}}$ by~$C^{\psi}_{1\mathrm{uPd}}$.

\begin{lemma} \label{P2}
For each $0\neq\psi\in\mathcal H$, one has
 \begin{equation}\label{integeA}
C^\psi_{1\mathrm{uPd}}=\bigcap_{l=1}^\infty\bigcap_{k=1}^\infty A^\psi_{1-1/(2l)-1/(2k)}\,,
\end{equation}
 where, for every $\alpha>0$,
\begin{equation*}
A^\psi_{\alpha}:=\bigcap_{n=0}^\infty\big\{T\in X\mid  \mathrm{\; for\;each}\; m,\;  \exists\; t>m \mathrm{\; with}\; t^{-\alpha/2}\,\Xi_{\psi}^T(t)>n\big\}\;.
\end{equation*}
\end{lemma}
\begin{proof} 
Fix~$\alpha\in(0,1)$ and let $\mathcal{H}_{\mathrm{a}\alpha{\mathrm{Pds}}}^{T}:=\{\zeta\in\mathcal H\mid \mu^T_{\zeta}\mathrm{\; is}\; \mathrm{a}\alpha{\mathrm{Pds}} \}$; by  Corollary~\ref{CPRo1},  
\[
\mathcal{H}_{\mathrm{a}\alpha{\mathrm{Pds}}}^{T}=\{\zeta\mid\dim_{\mathrm{P}}^+(\mu^T_{\zeta})\le\alpha\}\,.
\] 

By Proposition~\ref{AHC}, one has the inclusions $\mathcal{H}_{\alpha {\mathrm{Ps}}}^{T}\subset\mathcal{H}_{\mathrm{a}\alpha {\mathrm{Pds}}}^{T}\subset\mathcal{H}_{(\alpha+\var) {\mathrm{Ps}}}^{T}$; then, by Theorem~\ref{l21}, it follows that, for each~$T\in X$  and each $0<\varepsilon\le\min\{\alpha,(1-\alpha)/2\}$, 
\[
\big\{\zeta\mid \limsup_{t\rightarrow\infty}t^{-(\alpha-\var)/2}\,\Xi_{\zeta}^T(t)<\infty\big\}\subset \mathcal{H}_{\mathrm{a}\alpha {\mathrm{Pds}}}^{T}\setminus\{0\}\subset\big\{\zeta\mid \limsup_{t\rightarrow\infty}t^{-(\alpha+2\varepsilon)/2}\,\Xi_{\zeta}^T(t)<\infty\big\}.
\]  Hence, for fixed $0\neq\psi\in\mathcal{H}$ and $0<\varepsilon\le\min\{\alpha,(1-\alpha)/2\}$, one has
\begin{eqnarray*}
\bigcap_{n=0}^\infty\bigcap_{m=0}^\infty\bigcup_{t>m}\left\{T\in X\mid t^{-(\alpha+2\varepsilon)/2}\,\Xi_{\psi}^T(t)>n\right\}\subset\left\{T\in X\mid  \dim_{\mathrm{P}}^+(\mu_{\psi}^T)>\alpha\right\}\\
\subset\bigcap_{n=0}^\infty\bigcap_{m=0}^\infty\bigcup_{t>m}\left\{T\in X\mid t^{-(\alpha-\varepsilon)/2}\,\Xi_{\psi}^T(t)>n\right\};
\end{eqnarray*}
 that is, 
\begin{equation*}
A^\psi_{\alpha+2\varepsilon}\subset \{T\in X\mid\dim_{\mathrm{P}}^+(\mu_{\psi}^T)>\alpha\} \subset A^\psi_{\alpha-\varepsilon}\;.
\end{equation*} 

Finally, by replacing $\alpha$ by $\alpha-3\var$ and taking $\var=1/(8k)$, $k\ge 1$, and $\alpha=1-1/(2l)$, $l\ge 1$, one obtains
\begin{equation*}
\bigcap_{l=1}^\infty\bigcap_{k=1}^\infty A^\psi_{1-1/(2l)-1/(8k)} \subset C^\psi_{1\mathrm{uPd}}\subset \bigcap_{l=1}^\infty\bigcap_{k=1}^\infty A^\psi_{1-1/(2l)-1/(2k)}\;,
\end{equation*} since $C^\psi_{1\mathrm{uPd}}=\bigcap_{l=1}^\infty\bigcap_{k=1}^\infty \{T\in X\mid \dim_{\mathrm{P}}^+(\mu_{\psi}^T)>1-1/(2l)-3/(8k)\}$.  
\end{proof}

We remark that Lemma~\ref{P2} holds true for  restrictions to intervals~$(a,b)$. The choice $(a,b)=\mathbb R$ was just for simplicity. However, we will keep the interval in Theorem~\ref{prop1uPd} to deal with some subtleties there.

\begin{theorem} \label{prop1uPd}
Let $-\infty\le a<b\le+\infty$ and~$\psi\in\mathcal H$. Then, the set $C^{\psi;(a,b)}_{1\mathrm{uPd}}$ is a~$G_\delta$ set in~$X$.
\end{theorem}
\begin{proof}
If, for every $T\in X$, $\mu_{\psi}^T((a,b)\cap\cdot)=0$ (which is the case when, for every $T\in X$, $\supp(\mu_\psi^T)\cap(a,b)=\emptyset$), then $\dim_{\mathrm{P}}^+(\mu_{\psi;(a,b)}^T)=0$ and $C^{\psi;(a,b)}_{1\mathrm{uPd}}=\emptyset$ is a $G_\delta$ set in $X$. 

Otherwise, suppose that~$\xi=\xi(T,\psi,(a,b)):=P^T((a,b))\psi\neq 0$ for some~$T\in X$. Since, for bounded operators, strong  convergence implies strong resolvent convergence, which in turn is equivalent to strong dynamical convergence (see Theorem~10.1.8 in~\cite{Cesar}), it follows, for each $t\in\mathbb{R}$, that the mapping $X\ni T\mapsto \Xi_\psi^T(t)$ is continuous. 

Now, let $\mathcal{M}_+(I)$ represent the set of  positive finite Borel measures on the open interval~$I$ endowed with the vague topology; the continuity of the mapping $\mathcal{M}_+(\mathbb R)\ni\mu(\cdot)\mapsto \mu(I\cap\cdot)\in \mathcal{M}_+(I)$ (see~\cite{LSt} for details), combined with the continuity of $X\ni T\mapsto \Xi_\psi^T(t)$, implies that $X\ni T\mapsto \Xi_\xi^T(t)$ is also continuous. 

Observe that if there exist $T\in X$, $\xi\neq 0$, $t,x\in\mathbb{R}$ such that  $\int {\mathrm d}\mu_\xi^T(y)\,e^{-(x-y)^2t^2/4}>0$, then the continuity of $X\ni T\mapsto \int {\mathrm d}\mu_\xi^T(y)\,e^{-(x-y)^2t^2/4}$    implies that $\Xi_\xi^W(t)$ exists for every~$W$ in some neighbourhood of~$T$.

Therefore, one has, for every $n\ge 0$ and every $k,l\ge1$, that 
\[
\big\{T\in X\mid t^{-1/2+1/(4l)+1/(4k)}\,\Xi_{\xi}^T(t)> n\}
\] is an open subset of~$X$.  Now, relation~\eqref{integeA}, that is,
\begin{eqnarray*}
C^{\psi;(a,b)}_{1\mathrm{uPd}}=\bigcap_{l=1}^\infty\bigcap_{k=1}^\infty\bigcap_{n=0}^\infty\bigcap_{m=0}^\infty\bigcup_{t>m}\big\{T\in X\mid t^{-1/2+1/(4l)+1/(4k)}\,\Xi_{\xi}^T(t)> n\big\},
\end{eqnarray*} 
completes the proof.  
\end{proof}

\begin{corollary}
\label{R1}
Let $-\infty\le a<b\le+\infty$, $\psi\in\mathcal H$, and denote by  $(\mu_\psi^T)_{1\mathrm{Pc}}$ the $1$-packing continuous component of $\mu^T_{\psi }$. Suppose that  $C^{\psi;(a,b)}_{1\mathrm{uPc}}=\{T\in X\mid  (\mu_\psi^T)_{1\mathrm{Pc}}((a,b))\neq 0\}$ is dense in~$X$. Then,  the set~$C^{\psi;(a,b)}_{1\mathrm{uPd}}$ is generic in~$X$.
\end{corollary} 
\begin{proof}
Since, by Theorem~\ref{prop1uPd}, $C^{\psi;(a,b)}_{1\mathrm{uPd}}$ is a~$G_\delta$ set in~$X$, we just need to show that $C^{\psi;(a,b)}_{1\mathrm{uPd}}$ is dense. Suppose, then, that $(\mu_\psi^T)_{1\mathrm{Pc}}((a,b))>0$; thus, by Definitions~\ref{defDimMu} and~\ref{SCB}, $\dim_\mathrm{P}^+(\mu^T_{\psi;(a,b) })=1$, and therefore,~$C^{\psi;(a,b)}_{1\mathrm{uPc}}\subset C^{\psi;(a,b)}_{1\mathrm{uPd}}$. But now, since $C^{\psi;(a,b)}_{1\mathrm{uPc}}$ is dense, it follows that $C^{\psi;(a,b)}_{1\mathrm{uPd}}$ is also dense.  
\end{proof}

\OBSI \label{RIMP}
 A well-known fact about discrete Schr\"odinger operators in~$l^2(\mathbb{Z})$, with action~\eqref{LPo} and general real potentials~$(V_n)$, is the presence of a common set of cyclic vectors $\{\delta_{-1},\delta_0\}$. When the elements of the space~$X$ are of this type, the results stated in Corollary~\ref{R1} can be strengthened. Namely, if for $\zeta\in\{\delta_{-1},\delta_0\}$ the spectral measure $\mu^T_{\zeta;(a,b)}$ is 1Pd, then $\mu^T_{\psi;(a,b)}$ is 1Pd for every vector $\psi\ne0$ (since $P^T_{1\mathrm{Pd}}((a,b))=P^T((a,b))$ in this case), which implies that~$\{T\in X\mid\dim_{\mathrm{P}}(E)=1$ for some $E\subset\sigma(T)\cap(a,b)\}$  is a~$G_\delta$ set. 
\OBSF

Write $C_{1\mathrm{uPd}}:=\{T\in X_{\lambda,\nu}\mid$  $\dim_{\mathrm{P}}^+(\mu_{\delta_0}^T)=\dim_{\mathrm{P}}^+(\mu_{\delta_{-1}}^T)=1\}$. The inclusion $C_{1{\mathrm{uPd}}}\subset C_{{\mathrm{QB}}}$ (see the definition of the latter in the Introduction; this inclusion results from Proposition~\ref{GK} and the second inequality in~\eqref{bLIP}), together with Corollary~\ref{R1}, lead us to the following

\begin{proposition} 
Suppose that the hypotheses of Corollary~\ref{R1} are satisfied for~$\psi=\delta_0$ and $(a,b)=\mathbb{R}$. Then, $C_{{\mathrm{QB}}}$ is generic in~$X$.
\label{R2C}
\end{proposition}

\section{Proof of Theorem~\ref{LPO}}\label{PLPO}
\zerarcounters
We need the following 

\begin{theorem}[Theorem~1.1 in~\cite{DG}]
Suppose that~$\Omega$ is a Cantor group and that $\tau:\Omega\rightarrow\Omega$ is a minimal translation. Then, for a dense set of $g\in \CC(\Omega,\mathbb{R})$ and every $\kappa\in\Omega$, the spectrum of $H_{g,\tau}^\kappa$ is purely absolutely continuous.
\label{TDG}
\end{theorem}

\begin{proof}(Theorem~\ref{LPO})
Fix $\kappa\in\Omega$,  $\tau:\Omega\rightarrow\Omega$ a minimal translation of the Cantor group~$\Omega$, and consider in  $\CC(\Omega,\mathbb{R})$.

 Since, by Theorem~\ref{TDG},~$C_{\mathrm{ac}}^\kappa:=\{T\in X_\kappa\mid\sigma(T)$ is purely absolutely continuous$\}$ is dense in $X_\kappa$,  it follows that $C_{\mathrm{1uPd}}^\kappa\supset C_{\mathrm{ac}}^\kappa$ is also dense in~$X_{\kappa}$. Thus, by Corollary~\ref{R1} and Remark~\ref{RIMP}, we conclude that~$C_{\mathrm{1uPd}}^\kappa$ is generic in~$X_{\kappa}$.

The second assertion in the statement of the theorem follows from the inclusion~$C_{1{\mathrm{uPd}}}^\kappa\subset C_{{\mathrm{QB}}}^\kappa$ and Proposition~\ref{R2C}.
\end{proof}

\subsubsection*{Acknowledgment} This work was partially supported by the Brazilian Agency CNPq (Projeto Universal 41004/2014-8).

\noindent   Departamento de Matem\'atica, UFMG, Belo Horizonte,  MG, 30161-970 Brazil

\

\noindent  Departamento de Matem\'atica, UFSCar,  S\~{a}o Carlos, SP,  13560-970 Brazil

\end{document}